\documentclass[journal]{IEEEtran}

\usepackage{setspace}

\usepackage{graphics,
           psfrag,
           epsfig,
           amsthm,
           cite,
           amssymb,
           url,
           dsfont,
           subfigure,
           algorithm,
           algorithmic,
           balance,
           enumerate,
           color,
           setspace
}
\usepackage{amsmath}

\usepackage{epstopdf}

\newtheorem{definition}{Definition}

\newtheorem{lemma}{Lemma}

\newtheorem{theorem}{Theorem}

\newcommand{\eref}[1]{(\ref{#1})}
\newcommand{\sref}[1]{Section~\ref{#1}}
\newcommand{\appref}[1]{Appendix~\ref{#1}}
\newcommand{\fref}[1]{Figure~\ref{#1}}

\newcommand{\cref}[1]{Constraint~\ref{#1}}
\newcommand{\thref}[1]{Theorem~\ref{#1}}

\newcommand{\lref}[1]{Lemma~\ref{#1}}

\newcommand{\ignore}[1]{}

\IEEEoverridecommandlockouts
\begin{document}

\title{\vspace{-.5cm}Delivery Time Reduction for Order-Constrained Applications using Binary Network Codes}
\author{
\authorblockN{Ahmed Douik$^{\dagger}$, Mohammad S. Karim$^\ast$, Parastoo Sadeghi$^\ast$, and Sameh Sorour$^{\prime}$\\}%
\authorblockA{$^\dagger$California Institute of Technology (Caltech), California, United States of America \\
$^\ast$The Australian National University (ANU), Australia \\
$^\prime$King Fahd University of Petroleum and Minerals (KFUPM), Kingdom of Saudi Arabia \\
Email: $^\dagger$ahmed.douik@caltech.edu $^\ast$\{mohammad.karim,parastoo.sadeghi\}@anu.edu.au} $^{\prime}$samehsorour@kfupm.edu.sa
\vspace{-.8cm} }

\maketitle

\pagestyle{empty}
\thispagestyle{empty}

\IEEEoverridecommandlockouts

\begin{abstract}
Consider a radio access network wherein a base-station is required to deliver a set of order-constrained messages to a set of users over independent erasure channels. This paper studies the delivery time reduction problem using instantly decodable network coding (IDNC). Motivated by time-critical and order-constrained applications, the delivery time is defined, at each transmission, as the number of undelivered messages. The delivery time minimization problem being computationally intractable, most of the existing literature on IDNC propose sub-optimal online solutions. This paper suggests a novel method for solving the problem by introducing the delivery delay as a measure of distance to optimality. An expression characterizing the delivery time using the delivery delay is derived, allowing the approximation of the delivery time minimization problem by an optimization problem involving the delivery delay. The problem is, then, formulated as a maximum weight clique selection problem over the IDNC graph wherein the weight of each vertex reflects its corresponding user and message's delay. Simulation results suggest that the proposed solution achieves lower delivery and completion times as compared to the best-known heuristics for delivery time reduction.
\end{abstract}


\section{Introduction} \label{sec:int}

Various real-time applications in communication, e.g., cellular transmissions, video streaming, and satellite communications, require a considerable radio resources, i.e., bandwidth, transmission energy. To enhance the performance of such systems, network coding (NC), introduced in \cite{850663}, is a propitious solution that mixes the different information flows in the network \cite{6512065}. By achieving maximum information flow in a network \cite{4557282,4895447,1705002}, NC enables high-rate and reliable communications over fading channels.

While popular NC schemes, e.g., random linear network coding (RLNC) \cite{6195456,6714456,7070749}, focus only on achieving the maximum throughput in a network, they are not suitable for real-time applications of interest in this paper. For example, RLNC offers the optimal broadcast performance at the expense of a substantial decoding delay as decoding is possible only after the reception of a sufficient number of independently coded packets. However, many applications are time-critical and require in-order packet delivery as packets can be delivered to the applications only if all its preceding packets are decoded and delivered. Such applications include real-time scalable video streaming and cloud-enabled networks in which communications representing software commands need to be executed sequentially. A suitable NC technique to meet the aforementioned delay and message's order requirements is the instantly decodable network coding (IDNC) \cite{6882208,5502758,13051412,4476183,6570827,5072357,512548}

In IDNC, messages are encoded using the binary field $\mathds{F}_2$, i.e., messages are mixed using binary XOR. Such encoding field size allows efficient XOR-based decoding at the users by overcoming the expensive computations, e.g., large matrices inversion in RLNC. Such instant decodability property, not only reduces the decoding complexity but also enables the design of cost-efficient receivers.

For its aforementioned desirable properties, IDNC attracted a significant number of works. The authors in \cite{6882208,5502758,13051412} consider reducing the number of transmissions to complete the reception of the messages by all users. Such metric, known as the completion time, is desirable in applications without order constrains for its inverse relationship with the throughput. However, the metric is not suitable for order-constrained applications as out-of-order decoded messages are buffered but not delivered to the application. For real-time applications, the authors in \cite{4476183,6570827} propose serving the maximum number of users with any new message at each transmission. However, such approach is inefficient for order-constrained applications. For video streaming applications, reference \cite{5072357} suggests a video-aware packet selection algorithm that prioritizes messages based on their contribution to the overall video quality.

Consider a radio access network wherein a base-station is required to deliver a set of ordered messages to a set of users over independent erasure channels. The aim of this paper is to study the delivery time reduction problem in IDNC-based networks wherein the delivery time is incremented for each undelivered message irrespective of its decoding status. In an RLNC context, the authors in \cite{5437470,6214165} propose schemes that achieve the optimal asymptotic and a non-asymptotic satisfactory delivery time, respectively. Furthermore, the delivery time reduction problem considered in this paper is closely related to the concept developed in \cite{512548}. However, the authors in \cite{512548} formulate the optimal schedule that reduces the delivery time as a stochastic shortest path (SSP). For its high computational complexity, i.e., exponential in both the number of users and messages, they propose a simple packet selection heuristic.

This paper's main contribution is to propose a novel method for solving the delivery time reduction problem in IDNC-based networks. The delivery delay is first introduced as a measure of degradation as compared to optimal coding strategy. An expression characterizing the delivery time using the delivery delay is derived and used to approximate an anticipated version of the delivery time. Afterward, the problem is reformulated using the delay-dependent expression. The paper shows that the solution is equivalent to a maximum weight clique search over the IDNC graph wherein the weight of each vertex reflects its corresponding user and message's delay. Simulation results show appreciable performance gain and suggest that the proposed solution achieves a lower delivery and completion times as compared to the best-known heuristic \cite{512548,13051412,6570827} for delivery time reduction.

The rest of this paper is organized as follows: The system model and problem formulation are presented in \sref{sec:sys}. \sref{sec:del} introduces the delivery time approximation and reformulates the problem. The proposed solution is illustrated in \sref{sec:pro}. Before concluding in \sref{sec:con}, \sref{sec:sim} discusses the simulation results.

\section{System Model and Problem Formulation} \label{sec:sys}

\subsection{System Model and Parameters}

Consider the downlink of a radio access network with a single base-station (BS). The BS is required to deliver a set $\mathcal{M}$ of $M$ ordered messages\footnote{The term message, in this paper, denotes a generic packet that can represent a frame from a video stream, an executable instruction, and so on.} to a set $\mathcal{U}$ of $U$ users. Each user is interested in receiving all the messages of $\mathcal{M}$ in order. Out-of-order decoded messages are not delivered to the users' application layer but rather stored in their buffers. In other words, the $j$-th message, successfully received and decoded by the $u$-th user, is considered delivered to that user if and only if all previous messages $k < j$ are decoded and delivered.

\begin{figure}[t]
\centering
\includegraphics[width=0.8\linewidth]{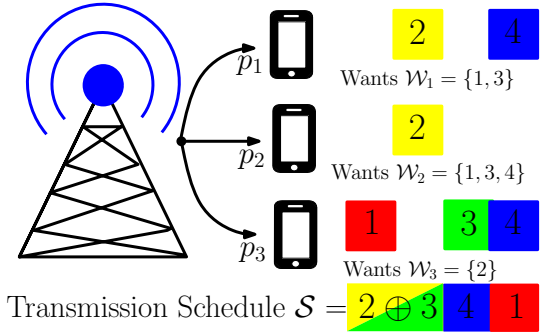}\\
\caption{A network composed of $3$ users and $4$ messages. The message combination $2 \oplus 3$ is instantly decodable for user $2$ but is out-of-order. The message is decoded and stored in the buffer resulting in a delivery time of $4$. The transmission schedule $\{2 \oplus 3,4,1\}$ results in an overall delivery time of $9$ and a completion time of $3$.}\label{fig:example}
\end{figure}

At each time slot, the BS broadcasts XOR combination of the source messages to the users. The transmission is subject to independent erasure at the different users. Let $p_u$ be the message erasure probability of the $u$-th user, assumed to be perfectly known to the BS and to remain constant during a single transmission. Each user that successfully receives a message sends an acknowledgment to the BS. This paper assumes perfect feedback reception that can be achieved through a proper choice of the modulation and frequency of the control channel. After each transmission, messages can be in the following sets of each user:
\begin{itemize}
\item The \emph{Has} set $\mathcal{H}_u$ including the messages received by the $u$-th user. In \fref{fig:example}, the Has set of user $3$ is $\mathcal{H}_3=\{1,3,4\}$.
\item The \emph{Wants} set $\mathcal{W}_u = \mathcal{M} \setminus \mathcal{H}_u$ including the messages wanted by the $u$-th user. In \fref{fig:example}, the Wants set of user $3$ is $\mathcal{W}_3=\{2\}$.
\item The \emph{Delivered} set $\mathcal{D}_u \subseteq \mathcal{H}_u$ including the messages delivered to the $u$-th user's application layer. In \fref{fig:example}, while the Has set of user $3$ is $\mathcal{H}_3=\{1,3,4\}$, its Delivered set is equal to $\mathcal{D}_3=\{1\}$.
\end{itemize}

Let $W_u^k \in \mathcal{W}_u$ denotes the $k$-th wanted message by the $u$-th user, e.g., in \fref{fig:example}, $W_3^1=2$ is the first wanted message by the third user, and $W_2^3=4$ is the third wanted message by the second user. The base-station exploits the diversity of Has and Wants sets of the different users to broadcast XOR combinations of the source messages. A message combination is instantly decodable for a user if it contains exactly one source message from its Wants set.

\subsection{Delivery and Completion Times}

This subsection defines two metrics, namely the completion time and the delivery time. First defined a schedule $\mathcal{S}$ as a collection of message combinations to be transmitted. For example, \fref{fig:example} represents a schedule $\{2\oplus 3,4,1\}$ containing $3$ message combinations. Further, let $\mathbf{S}$ be the set of all possible schedules.
\begin{definition}[Completion Time]
The completion time $\mathcal{C}(\mathcal{S})$ experienced after sending the schedule $\mathcal{S}$ is the number of transmissions required to deliver all messages to all users.
\end{definition}

The completion time reflects the minimum number of transmissions to complete the reception of the messages by all users, e.g., the completion time of the system and the schedule illustrated in \fref{fig:example} is $3$. However, such metric does not consider the order constraint of the messages and thus, is not suitable for order-constrained applications. To account for messages' order, the delivery time is defined as follows:
\begin{definition}[Delivery Time]
The delivery time $T_u(\mathcal{S})$ of the $u$-th user increases at each transmission by one unit for each undelivered message. In other words, the delivery time increases by $|\mathcal{M} \setminus \mathcal{D}_u|=M-W_u^1+1$ at each transmission before the completion time $\mathcal{C}(\mathcal{S})$.
\end{definition}

\begin{definition}[Overall Delivery Time]
The overall delivery time $\mathbf{T}(\mathcal{S})$ experienced after transmitting the schedule $\mathcal{S}$ is the sum of the delivery times of all users over all the transmissions until the completion time.
\end{definition}
The delivery time incorporates the messages' order by penalizing users for each undelivered message even if correctly decoded. For example, the overall delivery time of the system and the schedule illustrated in \fref{fig:example} is $9$. As transmissions order is of great importance, the delivery time is largely affected by it, e.g., while all three schedules $\{2\oplus 3,4,1\}$, $\{2\oplus 3,1,4\}$, and $\{1, 2\oplus 3,4\}$ in \fref{fig:example} achieve an equal completion time of $3$, their corresponding delivery times are $9$, $7$, and $10$, respectively.

\subsection{Problem Formulation}

The problem of finding the optimal schedule so as to minimize the delivery time in an IDNC-based system can be expressed as follows:
\begin{align}
\mathcal{S}^* = \arg \min_{\mathcal{S} \in \mathbf{S}} \mathbf{T}(\mathcal{S}) = \arg \min_{\mathcal{S} \in \mathbf{S}} \sum_{t=1}^{\mathcal{C}(\mathcal{S})} \sum_{u \in \mathcal{U}} T_u(t).
\label{eq:1}
\end{align}

It can readily be seen that finding the optimal schedule, i.e., the solution to the optimization problem \eref{eq:1}, is computationally intractable. Indeed, the dynamic nature of transmissions makes the problem anti-causal as the decision depends on future channel realizations and hence on future coding opportunities. Furthermore, the optimization is highly complex even for erasure free scenarios as it requires a search for all possible patterns of lost/received messages resulting in a complexity of order $2^{UM}$. The authors in \cite{512548} propose an SSP framework to reformulate the optimal schedule selection problem \eref{eq:1}. Given the high computational complexities of solving the SSP formulation, the characteristics of the SSP formulation are utilized to design a simple delivery time reduction heuristic. This paper suggests a novel method for solving the optimization problem \eref{eq:1} by introducing the delivery delay as a measure of degradation as compared to optimal coding strategy. Afterward, the problem is reformulated using a delivery time-delay dependent expression into a maximum weight clique selection problem in the IDNC graph.

\section{Delivery Time Reduction} \label{sec:del}

This section approximates the delivery time reduction problem by introducing the delivery delay. In particular, it first defines the delivery delay and derives an expression of the delivery time involving the delivery delay. It, then, proposes an anticipated version of the delivery time and approximates the minimum delivery time problem using such delivery delay dependent expression.

\subsection{Delivery Delay}

The delivery delay is introduced as a measure of degradation as compared to the optimal coding strategy. To define such delay, the following lemma characterizes the minimum delivery time of user for erasure free transmissions:
\begin{lemma}
Given any schedule $\mathcal{S}$, the minimum delivery time $\overline{W}_u$ for the $u$-th user is given by the following expression\footnote{The index $u$ in $\overline{W}_u$ is useful for studying scenarios wherein users initially hold a subset of $\mathcal{M}$, e.g., index coding problem \cite{1638566}. In such configuration, the minimum delivery time is different for different users based on their initially possessed packets and thus, $\overline{W}_u$ in \eref{eq:2} is also different for different users. However, the rest of the analysis holds.}:
\begin{align}
\overline{W}_u = \cfrac{M(M-1)}{2}.
\label{eq:2}
\end{align}
\label{lem1}
\end{lemma}

\begin{proof}
It can readily be seen that the minimum delivery time of the $u$-th user is achieved by transmitting the ordered messages sequentially. Assuming an erasure free scenario, the $t$-th transmission results in a successful delivery of the $t$-th message and an increase of $M-t$ in the delivery time. Therefore, the $M$ transmissions, required to complete the reception of all $M$ messages by the $u$-th user, translate in a minimum delivery time of $\overline{W}_u = \frac{M(M-1)}{2}$.
\end{proof}

The fundamental concept in defining the delivery delay is to measures the degradation as compared to the minimum delivery time. In other words, delivery time $T_u(\mathcal{S})$ experienced by the $u$-th user as a result of transmitting the schedule $\mathcal{S}$ is equal to the minimum delivery time $\overline{W}_u$ and the additional delivery delay $D_u(\mathcal{S})$ experienced by that user from schedule $\mathcal{S}$. Therefore, the delivery time and delay satisfy the following equation in erasure free scenarios:
\begin{align}
T_u(\mathcal{S}) = \overline{W}_u + D_u(\mathcal{S}).
\label{eq:3}
\end{align}

Given the constraint stated in \eref{eq:3}, the delivery delay is defined as follows:
\begin{definition}[Delivery Delay]
The delivery delay $D_u(t,\kappa)$ of the $u$-th user, at the $t$-th transmission, increases after the reception of the message combination $\kappa$ by the following quantity:
\begin{align}
D_u(t,\kappa) =
\begin{cases}
W_u^k - W_u^1 \hspace{0.5cm} &\text{if } \kappa \cap \mathcal{W}_u = W_u^k \\
M - W_u^1 + 1 \hspace{0.5cm} &\text{otherwise}
\end{cases}
\label{eq:4}
\end{align}
\end{definition}
In other words, the delivery delay increases by $W_u^k - W_u^1$ if the $k$-th wanted message by the $u$-th user is received. Otherwise, it increases by $M - W_u^1 + 1$. The following theorem characterizes the delivery time using a delivery delay dependent expression:
\begin{theorem}
Given any schedule $\mathcal{S}$, the delivery time $T_u(\mathcal{S})$ of the $u$-th user can be approximated by the following expression involving the delivery delay:
\begin{align}
T_u(\mathcal{S}) \approx \cfrac{\overline{W}_u + D_u(\mathcal{S})}{1-p_u}.
\label{eq:5}
\end{align}
\label{th1}
\end{theorem}

\begin{proof}
To demonstrate the theorem, the relationship is first established for an erasure free scenario, i.e., the delivery time is shown to satisfy the constraint defined in \eref{eq:3}. Such expression is shown while considering solely instantly decodable transmissions. The delay emanating from non-instantly decodable messages is then added to validate the expression proposed in \eref{eq:3}. Finally, the relationship is extended to the transmissions with erasure by approximating the additional delivery time caused by message erasure events. A complete proof can be found in \appref{app1}.
\end{proof}

The rest of this paper uses the approximation in \eref{eq:5} with equality as it indeed holds for erasure free scenarios, as shown in \eref{eq:3}, and for a large number of transmissions.

\subsection{Problem Reformulation}

As discussed in \sref{sec:sys}, the delivery time minimization problem is computationally intractable. Therefore, this subsection proposes approximating the problem by an online optimization problem involving an anticipated version of the delivery time.

Let $T_u(t)$ be the anticipated delivery time of the $u$-th user at the $t$-th transmission. Such quantity approximates the expected delivery time of the $u$-th user at the $t$-th transmission and can be defined as follows:
\begin{align}
T_u(t) = \cfrac{\overline{W}_u + D_u(t)}{1-p_u},
\end{align}
where $D_u(t)$ is the cumulative delivery delay experienced by the $u$-th user from the first until the $t$-th transmission. It can be seen that the anticipated delivery time $T_u(t)$ is equal to the individual delivery time $T_u(\mathcal{S})$ if the $u$-th user does not experience any additional delivery delay in future transmissions.

This subsection, now, proposes approximating the delivery time reduction problem \eref{eq:1} by the following online optimization problem over the message combination $\kappa$:
\begin{align}
\kappa^* = \arg \min_{\kappa \in \mathcal{P}(\mathcal{M})} \sum_{u \in \mathcal{U}} T_u(t,\kappa),
\label{eq:7}
\end{align}
where $\mathcal{P}(\mathcal{M})$ represents the power-set of the set of messages $\mathcal{M}$.

\section{Proposed Solution} \label{sec:pro}

This section suggests finding the optimal message combination that minimizes the expected delivery time, i.e., online delivery time reduction problem \eref{eq:7}. To represent, in one unified framework, all possible message combinations and the users to whom each message combination is intended, this section first presents the IDNC graph. Afterward, the optimization problem \eref{eq:7} is reformulated as a maximum weight clique selection problem wherein the weight of each vertex in the IDNC graph represents the delivery delay of its user and message combination.

The IDNC graph $\mathcal{G}(\mathcal{V},\mathcal{E})$ is a tool introduced in \cite{5683677} to represent all feasible message combinations and the users to whom the transmission is instantly decodable. The set of vertices is constructed by generating a vertex $v \in \mathcal{V}$ for each couple of user and wanted message, i.e., a vertex $v_{um}$ is produced for each user $u \in \mathcal{U}$ and wanted message $m \in \mathcal{W}_u$. An edge $e \in \mathcal{E}$ is generated for each two vertices $v_{um}$ and $v_{u^\prime m^\prime}$ when the combination of the messages $m$ and $m^\prime$ is instantly decodable to both users $u$ and $u^\prime$. From the instant decodability constraint of IDNC, it can readily be seen that two vertices $v_{um}$ and $v_{u^\prime m^\prime}$ are adjacent if one of the following two options is true:
\begin{itemize}
\item $m = m^\prime$: The same message is requested by two different users and thus the combination is instantly decodable for both users.
\item $m \in \mathcal{H}_{u^\prime}$ and $m^\prime \in \mathcal{H}_{u}$: Both users $u$ and $u^\prime$ can XOR the combination $m \oplus m^\prime$ to retrieve the messages $m$ and $m^\prime$, respectively.
\end{itemize}

Given the IDNC graph formulation above, the following theorem characterizes the solution to the delivery time reduction problem \eref{eq:7}:
\begin{theorem}
The optimal message combination the base-station can generate at the $t$-th transmission so as to reduce the anticipated delivery time proposed in \eref{eq:7} is the maximum-weight clique in the IDNC graph wherein the weight of a vertex $v_{um}$ is defined by:
\begin{align}
w(v_{um}) = \cfrac{M - m + 1}{1-p_u}.
\label{eq:6}
\end{align}
\label{th2}
\end{theorem}

\begin{proof}
To show this theorem, the optimal message combination $\kappa$ is first expressed as a function of the targeted users. Afterward, using the bijection between the set of maximal cliques in the IDNC graph and the set of message combinations and targeted users, the message selection is expressed as a maximal clique search over the graph. To conclude the proof, the weight of the vertices is demonstrated to represent the objective function of \eref{eq:7}. A complete proof can be found in \appref{app2}.
\end{proof}

\section{Simulation Results} \label{sec:sim}

This section presents the simulation results assessing the performance of the proposed solution, denoted by \emph{minimum average delivery time} (Min-ADT), in the downlink of a radio access network. A large number of iteration is performed and the mean value of the delivery time, denoted by average delivery time, is presented. The number of users, messages, and erasure probabilities are variable in the simulations so as to show the performance of the different algorithms in various scenarios. The proposed solution is compared, in terms of delivery and completion times, against the following algorithms:
\begin{itemize}
\item The delivery time reduction algorithm introduced in \cite{512548}. The heuristic scheme, denoted by ``SSP-H", is based on the properties of the SSP formulation.
\item The completion time reduction algorithm introduced in \cite{13051412}. The heuristic reduces the completion time while ignoring the messages' order in the selection process.
\item The maximum clique selection algorithm introduced in \cite{6570827}. The algorithm selects the maximum clique over the IDNC graph and targets the maximum number of users with a new message for each transmission.
\end{itemize}

\begin{figure}[t]
\centering
\includegraphics[width=0.9\linewidth]{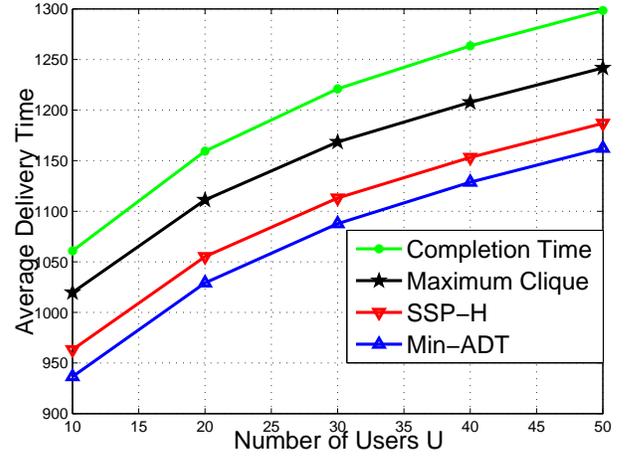}\\
\caption{Average delivery time versus the number of users $U$ for a network composed of $M=30$ messages and an average erasure probability $P=0.25$.}\label{fig:MDD}
\end{figure}

\begin{figure}[t]
\centering
\includegraphics[width=0.9\linewidth]{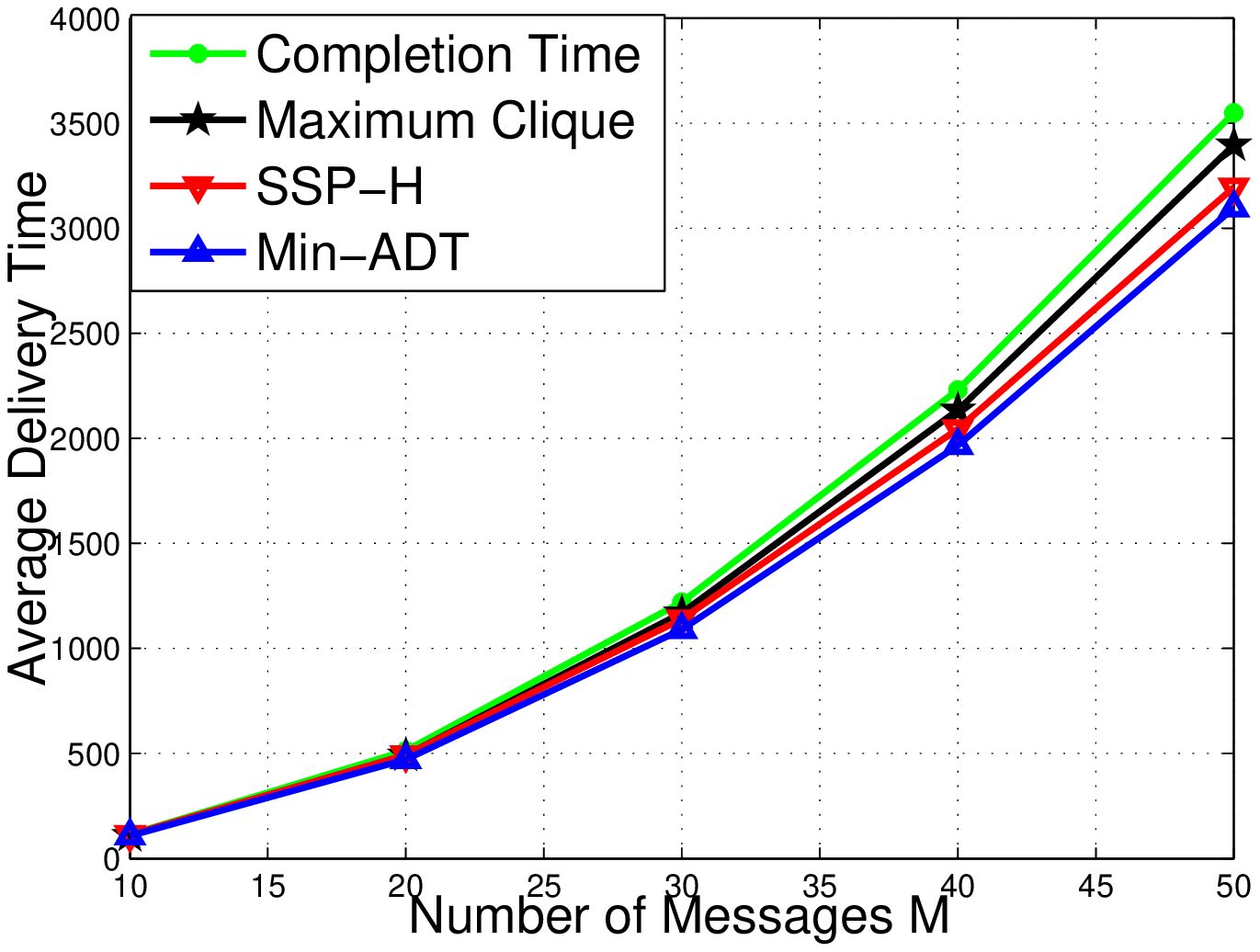}\\
\caption{Average delivery time versus the number of messages $M$ for a network composed of $U=30$ users and an average erasure probability $P=0.25$.}\label{fig:NDD}
\end{figure}

\begin{figure}[t]
\centering
\includegraphics[width=0.9\linewidth]{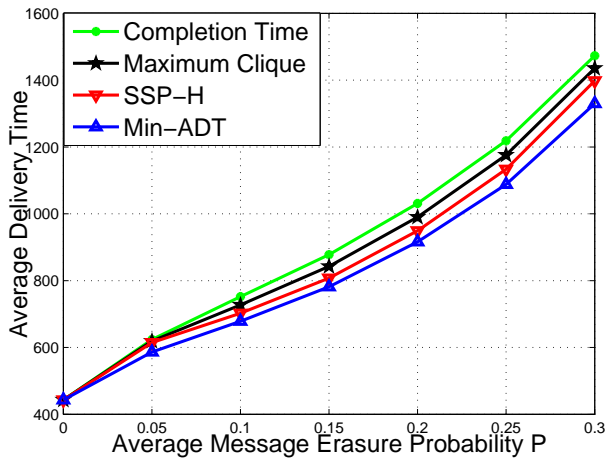}\\
\caption{Average delivery time versus the erasure probability $P$ for a network composed of $U=30$ users and $M=30$ messages.}\label{fig:PDD}
\end{figure}

\begin{figure}[t]
\centering
\includegraphics[width=0.9\linewidth]{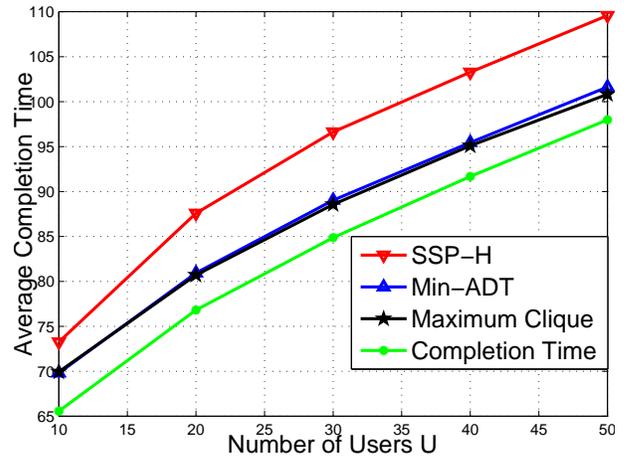}\\
\caption{Average completion time versus the number of users $U$ for a network composed of $M=30$ messages and an average erasure probability $P=0.25$.}\label{fig:MCT}
\end{figure}

\fref{fig:MDD} illustrates the delivery time achieved by the various algorithms versus the number of users $U$ for a network composed of $M=30$ messages and an average erasure probability $P=0.25$. The figure suggests that the proposed solution largely outperforms the three other schemes by achieving a smaller delivery time. In other words, the proposed solution achieves quickly in-order message delivery to the application layers of the users. For a fixed number of messages, the performance of the proposed algorithm degrades as the number of users increases. This can be explained by the fact that the delivery time approximation becomes less accurate as the number of users increases in the network. As both the completion time algorithm and the maximum clique solution do not consider the messages' order in the selection process, they poorly perform in reducing the delivery time.

\fref{fig:NDD} depicts the delivery time performances of the different algorithms versus the number of messages $M$ for a network composed of $U=30$ users and an average erasure probability $P=0.25$. The proposed solution achieves a lower delivery time for all number of messages. Moreover, the performance gap increases as the total number of messages in the network increases. This can be explained by the fact that as the number of messages increases, the coding opportunities generally increases. Such coding opportunities come in favor of the proposed solution as it efficiently selects the message combination by incorporating the delivery delay in the vertices' weigh as expressed \eref{eq:6}.

\fref{fig:PDD} shows the delivery time against different average erasure probabilities for a network composed of $U=30$ users and $M=30$ messages. As expected, the proposed solution outperforms other three algorithms, especially as the erasure probability increases. This can be explained by the fact that, as the erasure probability increases, the estimation of the delivery time becomes more accurate. In fact, as shown in \thref{th2}, the delivery time is approximated using the average number of erased transmissions. For large erasure probabilities, such approximation holds by the law of large number, resulting in a better performance of the proposed solution as compare to other schemes.

Finally, \fref{fig:MCT} presents the completion time achieved by different algorithms against the number of users $U$ for a network composed of $M=30$ messages and an average erasure probability $P=0.25$. As explained in \sref{sec:sys}, the completion time reflects the minimum number of transmissions so as to complete the reception of all messages to all users regardless of the messages' order. The figure clearly shows that the proposed solution, unlike SSP-H, presents a reasonable degradation in the completion time against the best-known completion time reduction heuristic while preserving the benefits of the delivery time reduction. The completion time reduction performance of the proposed solution is closely related to the Maximum Clique algorithm that serves the possible largest number of users with any new message in each transmission. In fact, the proposed solution, while reducing the delivery time, targets a large number of users.

\section{Conclusion} \label{sec:con}

Consider a radio access network wherein a base-station is required to deliver a set of order-constrained messages to a set of users over independent erasure channels. This paper proposes a novel method for solving the delivery time reduction problem for order-constrained applications using instantly decodable network coding. The notion of delivery delay is introduced as a measure of degradation against the optimal coding strategy in an erasure free scenario. The delivery time is, then, approximated by an anticipated version that incorporates the delivery delay. The delivery time reduction problem is reformulated using the delivery time-delay dependent expression and shown to be equivalent to a maximum weight clique selection problem over the IDNC graph. Simulation results show that the proposed solution provides an appreciable performance as compared to the best-known delivery time reduction heuristic. In addition to delivery time reduction benefit, the results further suggest that the proposed solution achieves a tolerable completion time degradation as compared to the best-known order unconstrained completion time reduction heuristic.

\appendices

\numberwithin{equation}{section}

\section{Proof of \thref{th1}}\label{app1}

The proof of this theorem goes as follows. The delivery time-delay expression is first derived for erasure free scenarios. In other words, the relationship is first established for $p_u=0, \ \forall \ u \in \mathcal{U}$. Afterward, the relationship is extended to transmissions with erasure by approximating the additional delivery time resulting from message erasures. The delivery time-delay expression \eref{eq:3} is demonstrated for a special schedule containing solely instantly decodable messages. Finally, it is extended to an arbitrary schedule by adding delay caused by non-instantly decodable transmissions.

Let $\mathcal{S}$ be a special transmission schedule containing only instantly decodable messages for the $u$-th user. Therefore, each transmission brings a new message to the user. Given that the user wants $M$ messages, it can easily be concluded that the schedule $\mathcal{S}$ contains $M$ transmissions. Hence, the schedule is a permutation of the $M$ messages. From its definition, the delivery time of the $u$-th user can be expressed as follows:
\begin{align}
T_u(\mathcal{S}) = \sum_{t=1}^{M-1} \left( M - W_u^1(t) + 1 \right),
\label{eq:app12}
\end{align}
where $W_u^1(t)$ is the first wanted message by the $u$-th user at the $t$-th transmission. Note that the last transmission in the schedule $\mathcal{S}$ brings the last instantly decodable message for the $u$-th user. Therefore, the user does not experience any delivery time increase in the last transmission.

Let $\kappa(t) \in \mathcal{M}$ be the message of the $t$-th transmission. From the analysis above, it can be concluded that $\bigcup_{t=1}^M \kappa(t) = \mathcal{M}$. Therefore, the delivery time of the $u$-th user in \eref{eq:app12} is given by the following expression:
\begin{align}
T_u(\mathcal{S}) &= \sum_{t=1}^{M-1} \left( M - W_u^1(t) + 1 + \kappa(t) - \kappa(t)\right) \nonumber \\
&= \sum_{t=1}^{M-1} \left(M - \kappa(t) + 1\right) + \sum_{t=1}^{M-1} \left(\kappa(t) - W_u^1(t) \right).
\label{eq:ap11}
\end{align}

The first term in \eref{eq:ap11} represents the minimum delivery time illustrated in \lref{lem1}, i.e., $\sum_{t=1}^{M-1} M - \kappa(t) + 1 = \overline{W}_u$. Therefore, to show that the expression \eref{eq:3} holds, it is sufficient to show that the second term represents the delivery delay $D_u(\mathcal{S})$. Given that all transmissions are instantly decodable in the schedule $\mathcal{S}$ and provided expression \eref{eq:4}, it can be inferred that $\kappa(t) - W_u^1(t) = D_u(t,\kappa)$. Therefore, the delivery time of the $u$-th user is:
\begin{align}
T_u(\mathcal{S}) &= \overline{W}_u + \sum_{t=1}^{M-1} D_u(t,\kappa) = \overline{W}_u + D_u(\mathcal{S}).
\end{align}

Now, consider an arbitrary schedule $\mathcal{S}$ with both instantly and non-instantly decodable messages. For the $u$-th user, the schedule can be decomposed into two schedules: the first $\mathcal{S}_p$ containing all instantly decodable transmissions for the $u$-th user and the second $\mathcal{S}_s$ containing all non-instantly decodable transmissions for that user. From the previous analysis in \eref{eq:app12} and \eref{eq:ap11}, the delivery time of the $u$-th user can be written as follows:
\begin{align}
T_u(\mathcal{S}) &= T_u(\mathcal{S}_p) + T_u(\mathcal{S}_s)\nonumber \\
&= \overline{W}_u + D_u(\mathcal{S}_p) + T_u(\mathcal{S}_s) \nonumber \\
&= \overline{W}_u + D_u(\mathcal{S}_p) + \sum_{t \in \mathcal{S}_s} M - W_u^1(t) +1.
\end{align}

Given that all transmission in $\mathcal{S}_s$ are non-instantly decodable for the $u$-th user, the first wanted message $W_u^1(t)$ remains unchanged in each of those transmissions. Therefore, for each non-instantly decodable message combination $\kappa$, the following equality holds: $M - W_u^1(t) +1 = D_u(t,\kappa)$. With this result, the delivery time of the $u$-th user can be defined as:
\begin{align}
T_u(\mathcal{S}) &= \overline{W}_u + D_u(\mathcal{S}_p) + \sum_{t \in \mathcal{S}_s} D_u(t,\kappa) \nonumber \\
&= \overline{W}_u + D_u(\mathcal{S}_p) + D_u(\mathcal{S}_s) \nonumber \\
&= \overline{W}_u + D_u(\mathcal{S}).
\end{align}

Having established the expression given in \eref{eq:3}, the analysis is now extended to the message erasure scenarios by approximating the delivery delay caused from all erased messages in schedule $\mathcal{S}$. For a schedule $\mathcal{S}$, let $\mathcal{E}_u(\mathcal{S})$ be the additional delivery time caused by the erased messages at the $u$-th user. Now, the delivery time is defined in terms of the minimum delivery time, the delivery delay, and the erased transmissions as follows:
\begin{align}
T_u(\mathcal{S}) = \overline{W}_u + D_u(\mathcal{S}) + \mathcal{E}_u(\mathcal{S})
\label{eq:app14}
\end{align}

Let $X_u(t)$ be a Bernoulli random variable indicating ($X_u(t)=1$) that the $t$-th transmission is erased at the $u$-th user. The additional delivery time caused by erased messages in schedule $\mathcal{S}$ can be expressed as:
\begin{align}
\mathcal{E}_u(\mathcal{S}) = \sum_{t=1}^{|\mathcal{S}|-1} (M-W_u^1(t)+1)X_u(t).
\label{feee}
\end{align}

Similar to the expression in \eref{eq:app12}, the last transmission is instantly decodable for the $u$-th user and thus, no delivery time increase occurs from that transmission. The expected value of the additional delivery delay caused by erased messages at the $u$-th user is:
\begin{align}
\mathds{E}[\mathcal{E}_u(\mathcal{S})] &= \mathds{E}[\sum_{t=1}^{|\mathcal{S}|-1} (M-W_u^1(t)+1)X_u(t)] \nonumber \\
&= \sum_{t=1}^{|\mathcal{S}|-1} (M-W_u^1(t)+1)\mathds{E}[X_u(t)] \nonumber \\
&= p_u \sum_{t=1}^{|\mathcal{S}|-1} (M-W_u^1(t)+1) \nonumber \\
&= p_u \mathcal{T}_u(\mathcal{S})
\label{eq:app13}
\end{align}

This paper proposes approximating the additional delivery time \eref{feee} by its average value in \eref{eq:app13}, i.e., $\mathcal{E}_u(\mathcal{S}) \approx \mathds{E}[\mathcal{E}_u(\mathcal{S})]$. Substituting and rearranging the terms of the expression \eref{eq:app14} gives the desired result:
\begin{align}
T_u(\mathcal{S}) \approx \cfrac{\overline{W}_u + D_u(\mathcal{S})}{1-p_u}.
\end{align}

\section{Proof of \thref{th2}}\label{app2}

The steps of the proof are the followings. The optimal message combination $\kappa^*$ is first expressed as a function of the targeted users. Afterward, using the bijection between the set of maximal cliques in the IDNC graph and the set of message combinations and targeted users, the message selection is expressed as a maximal clique search over the graph. To conclude the proof, the weight of the vertices is demonstrated to represent the objective function of \eref{eq:7}.

To begin with, note that the delivery time and the delay experienced in the previous transmissions are not function of the message combination $\kappa$ at the $t$-th transmission. Hence, the optimization problem \eref{eq:7} can be simplified in terms of the delivery delay and the erasure probabilities as follows:
\begin{align}
\kappa^* &= \arg \min_{\kappa \in \mathcal{P}(\mathcal{M})} \sum_{u \in \mathcal{U}} T_u(\kappa) \nonumber \\
&= \arg \min_{\kappa \in \mathcal{P}(\mathcal{M})} \sum_{u \in \mathcal{U}} \cfrac{\overline{W}_u + D_u(t,\kappa) + D_u(t-1)}{1-p_u} \nonumber \\
&= \arg \min_{\kappa \in \mathcal{P}(\mathcal{M})} \sum_{u \in \mathcal{U}} \cfrac{D_u(t,\kappa)}{1-p_u}
\label{eqw}
\end{align}

Let $U_w$ be the set of users with non-empty Wants set and $\tau(\kappa)$ be the set of targeted users that can instantly decode a new message from the combination $\kappa$. From the definition of the delivery delay in \eref{eq:4}, a targeted user $u$ experiences $W_u^k-W_u^1$ unit of delay increase, wherein $k$ is the new message of the $u$-th user in the combination $\kappa$. A non-targeted user $u$ by the combination $\kappa$ experiences $M-W_u^1+1$ unit of delay increase. Therefore, the optimal message combination in \eref{eqw} can be reformulated as follows:
\begin{align}
&\kappa^* = \arg \min_{\kappa \in \mathcal{P}(\mathcal{M})} \sum_{u \in \mathcal{U}} \cfrac{D_u(t,\kappa)}{1-p_u} \nonumber \\
&= \arg \min_{\kappa \in \mathcal{P}(\mathcal{M})} \sum_{u \in \tau(\kappa)}\cfrac{ W_u^k-W_u^1}{1-p_u} + \sum_{u \in U_w \setminus \tau(\kappa)} \cfrac{M-W_u^1+1}{1-p_u} \nonumber \\
&= \arg \max_{\kappa \in \mathcal{P}(\mathcal{M})} \sum_{u \in \tau(\kappa)} \cfrac{M-W_u^1+1}{1-p_u} - \sum_{u \in \tau(\kappa)}\cfrac{ W_u^k-W_u^1}{1-p_u} \nonumber \\
&= \arg \max_{\kappa \in \mathcal{P}(\mathcal{M})} \sum_{u \in \tau(\kappa)} \cfrac{M-W_u^k+1}{1-p_u}.
\end{align}

According to the analysis performed in \cite{6570827}, there exists a one-to-one mapping between the set of feasible message combinations and the set of maximal cliques in the IDNC graph. Let $\mathbf{C}$ be the set of maximal cliques in the IDNC graph. The optimal message combination can be expressed as follows:
\begin{align}
\kappa^* &= \arg \max_{\kappa \in \mathcal{P}(\mathcal{M})} \sum_{u \in \tau(\kappa)} \cfrac{M-W_u^k+1}{1-p_u} \nonumber \\
&= \arg \max_{ C \in \mathbf{C}} \sum_{v_{um} \in C} \cfrac{M-W_u^k+1}{1-p_u},
\end{align}
where $W_u^k$ is the intended message to the $u$-th user in the transmission of the maximal clique $C$. By construction of the IDNC graph, a vertex $v_{um}$ translates that the $u$-th user wants the $m$-th message. Given that a maximal clique is instantly decodable for all the users represented by that clique, the wanted message $W_u^k$ inducing vertex $v_{um}$ is the message $m$. Therefore, the optimization problem \eref{eq:6} can be expressed as:
\begin{align}
\max_{ C \in \mathbf{C}} \sum_{v_{um} \in C} \cfrac{M-m+1}{1-p_u} = \max_{ C \in \mathbf{C}} \sum_{v_{um} \in C} w(v_{um}).
\end{align}

Therefore, the optimal message combination is the maximum weight clique in the IDNC graph wherein the weights of vertices are defined in \eref{eq:6}.

\bibliographystyle{IEEEtran}
\bibliography{citations}

\end{document}